\documentclass[a4paper,UKenglish,cleveref,thm-restate]{lipics-v2021}

\usepackage{nicefrac}
\usepackage{mathtools}
\usepackage{xspace}
\usepackage{complexity}
\usepackage[normalem]{ulem}
\usepackage{siunitx}
\usepackage[table]{xcolor}

\graphicspath{{./figures/}}

\crefname{conjecture}{Conjecture}{Conjectures}

\newcommand{\flipSeq}{\ensuremath{\mathcal{S}}\xspace}
\newcommand{\trace}{\ensuremath{\operatorname{trace}}}

\title{Structural Properties of Shortest Flip Sequences Between Plane Spanning Trees}
\titlerunning{Structural Properties of Shortest Flip Sequences Between Plane Spanning Trees}

    \author{Oswin Aichholzer}{Institute of Algorithms and Theory, Graz University of Technology, Austria}{oswin.aichholzer@tugraz.at}{https://orcid.org/0000-0002-2364-0583}{}
    \author{Joseph Dorfer}{Institute of Algorithms and Theory, Graz University of Technology, Austria}{joseph.dorfer@tugraz.at}{https://orcid.org/0009-0004-9276-7870}{Austrian Science Fund (FWF) 10.55776/DOC183.}
    \author{Peter Kramer}{Department of Computer Science, TU Braunschweig, Germany}{kramer@ibr.cs.tu-bs.de}{https://orcid.org/0000-0001-9635-5890}{Fellowship of the German Academic Exchange Service~(DAAD).}
    \author{Christian Rieck}{Institute of Mathematics, University of Kassel, Germany}{christian.rieck@mathematik.uni-kassel.de}{https://orcid.org/0000-0003-0846-5163}{Funded by the Deutsche Forschungsgemeinschaft (DFG, German Research Foundation) -- 522790373.}
    \author{Birgit Vogtenhuber}{Institute of Algorithms and Theory, Graz University of Technology, Austria}{birgit.vogtenhuber@tugraz.at}{https://orcid.org/0000-0002-7166-4467}{}
    \authorrunning{O.~Aichholzer, J.~Dorfer, P.~Kramer, C.~Rieck, and B.~Vogtenhuber}
    \Copyright{Oswin Aichholzer, Joseph Dorfer, Peter Kramer, Christian Rieck, and Birgit Vogtenhuber}



\hideLIPIcs
\nolinenumbers
\acknowledgements{This work was initiated at the Reconfiguration Week hosted at TU Graz in March~2025. We thank the organizers and all participants for fruitful discussions and the very  good~atmosphere.}

\EventEditors{}
\EventNoEds{1}
\EventLongTitle{}
\EventShortTitle{}
\EventAcronym{}
\EventYear{}
\EventDate{}
\EventLocation{}
\EventLogo{}
\SeriesVolume{}
\ArticleNo{}

\begin{document}
    \maketitle

    \begin{abstract}
        We study the problem of reconfiguring non-crossing spanning trees on point sets in the plane in convex position. A \emph{flip} is a reconfiguration step in which one edge of a tree is replaced with another, yielding again a non-crossing spanning tree. For a given point set $P$, the \emph{flip graph} has a vertex for each non-crossing spanning tree on $P$, and an edge between any two trees that differ by a single flip. The \emph{flip distance} between two trees is the length of a shortest path between them in the flip graph, that is, the minimum number of flips needed to transform one tree into the other.

        We study structural properties of shortest flip sequences. Over the years, several conjectures have been proposed regarding these sequences. The folklore \emph{happy~edge conjecture} suggests that any edge shared by both the initial and target tree is never flipped in a shortest flip sequence. Moreover, the \emph{parking conjecture} claims that in shortest flip sequences,
        every edge is either flipped to a target edge or parked on the convex hull. Notably, the happy edge conjecture is implied by the parking conjecture. Furthermore, it has been conjectured that an edge is never \emph{reparked}, that is, any edge resulting from a flip either is a target edge or is subsequently flipped to a target edge. Essentially, all recent flip algorithms rely on these three conjectures and the properties they imply.

        In this paper, we disprove both the parking and the reparking conjecture. In particular, we show that a flip sequence obeying the parking property sometimes needs linearly more flips than a shortest flip sequence, and exhibit examples where any shortest flip sequence flips some edge multiple times. In contrast, we identify several conditions under which the parking and reparking properties hold; these imply that the reparking property holds for compatible flips, and may be exploited by algorithms to obtain short flip sequences.
    \end{abstract}


\pagebreak
\section{Introduction}
\label{sec:introduction}

Reconfiguration problems ask whether one feasible state of a system can be transformed into another through a sequence of allowed moves, a perspective made intuitive by classic puzzles such as the Rubik's cube or the 15-puzzle.
In the former, for instance, a single rotation of a face, either \ang{90} or \ang{180}, constitutes an elementary move, and sequences of such face rotations navigate a vast state space of about $4.3\cdot 10^{19}$ reachable configurations.
The structure of this state space, that is, how configurations are connected by allowed moves, is captured formally by the respective reconfiguration graph, where vertices represent configurations and edges correspond to single moves.
Remarkably, despite its enormous size, the reconfiguration graph of the Rubik's cube has diameter $20$, so any scrambled cube can be solved in at most~$20$ face rotations.
This viewpoint extends far beyond puzzles: in combinatorics, (computational) geometry, and theoretical computer science in general, reconfiguration graphs are a tool for studying families of objects that can be transformed into one another.

Another natural and illustrative example of geometric reconfiguration is the reconfiguration of triangulations of a convex point set.
In this context, a flip is a local move that replaces one diagonal of a convex quadrilateral formed by two adjacent triangles with the other diagonal, thereby obtaining a new triangulation.
The reconfiguration graph (also called flip graph) is the $1$-skeleton of the associahedron.
These triangulations are in bijection with binary trees~\cite{SleatorTT86}, and each flip corresponds exactly to a rotation in the associated tree.
Hence, the number of flips needed to transform one triangulation into another is the same as the number of rotations required to reorganize the corresponding binary trees.
Since binary search trees are a fundamental data structure, understanding and minimizing these rotations is crucial for improving the efficiency of algorithms that rely on tree operations.
Even though it is known that the diameter of the associahedron is equal to $2n-6$ (for $n\geq 11$), it has remained an open problem whether there exists a polynomial-time algorithm for calculating the rotation distance~\cite{POURNIN201413,SleatorTT86}.
In a very recent preprint, Dorfer~\cite{dorfer2026flipdistance} answers this question in the negative.

\medskip
In this paper, we consider the reconfiguration of non-crossing spanning trees on point sets in convex position.
A \emph{flip} in a non-crossing spanning tree removes one edge from the tree and adds another edge, resulting again in a non-crossing spanning tree.
The corresponding \emph{flip graph} has as vertices all non-crossing spanning trees on a convex set~$P$ of $n$ points, and two vertices are connected by an edge precisely when the associated trees differ in a single flip.
A~\emph{flip sequence} between two trees is any path in the flip graph from the vertex representing one tree to that representing the other tree, and the \emph{flip distance} is the length of a shortest such path. See~\cref{fig:intro} for the flip graph of plane spanning trees on $4$ points in non-convex~position.

\begin{figure}[htb]
    \centering
    \includegraphics[scale=0.8]{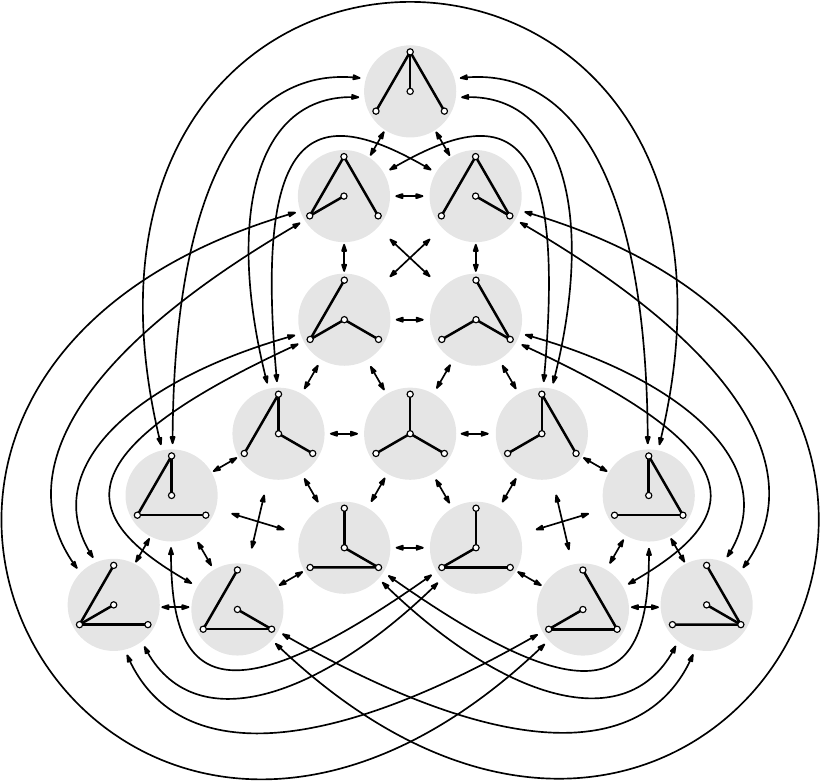}
    \caption{Flip graph of plane spanning trees on $n=4$ points in non-convex position. Puzzle:~Can you find all 120 different possibilities how the diameter of this flip graph can be realized?}
    \label{fig:intro}
\end{figure}

Recent research has focused on (1) the \emph{diameter} of the flip graph, that is, the maximum flip distance between any pair of trees on convex point sets and (2) algorithmic aspects of the computation of the flip distance.
For over two decades, the benchmark results for bounds on the diameter were an upper and a lower bound of $2n-4$~\cite{AVIS199621} and $\lfloor\frac{3n}{2}\rfloor-5$~\cite{HERNANDO199951}, respectively.
The upper bound on the diameter was later improved to $2n-\log(n)$~\cite{AichholzerBBDDKLLTU24}, and shortly after to $2n-\sqrt{n}$~\cite{bousquet2023notes}.
In~\cite{bousquet2025reconfiguration}, the longstanding $2n$ barrier was surpassed by proving an upper bound on the diameter of $\approx1.96n$.
The currently best known bounds are a lower bound of $(\frac{11}{7}-o(1))n$ by Bjerkevik, Dorfer, Kleist, Ueckerdt, and Vogtenhuber~\cite{bjerkevik2026flippingHardness}, and an upper bound of $\frac{5n}{3}-3$ by Bjerkevik, Kleist, Ueckerdt, and Vogtenhuber~\cite{bjerkevik2025flipping}.

One potentially useful tool for tackling the above two questions is the set of edges common to the start and target tree, known as \emph{happy edges}.
All other edges are referred to as \emph{unhappy}.
Building on this concept, the following long-standing conjecture is considered folklore.

\begin{conjecture}[Happy edge conjecture]
    \label{conj:happy}
    Let $T_I$ and $T_F$ be two plane spanning trees, and let $H=T_I\cap T_F$ be the set of all happy edges.
    Then there exists a shortest flip sequence~$T_I=T_0,T_1,\ldots,T_{k-1},T_k=T_F$ such that $H\subset T_i$ for every $i\in[0,k]$.
\end{conjecture}

Most of the aforementioned improvements to the diameter bounds are based on paths that respect this conjecture, that is, they never flip happy edges.
Specifically, the authors of~\cite{AichholzerBBDDKLLTU24} present their $2n-\log(n)$ bound as follows:
\begin{quote}
    \emph{``...happy edges do not flip. An unhappy edge of $T_I$ flips once or twice, and if twice, it flips to the convex hull and then to an edge of $T_F$.''}
\end{quote}

Moreover, they define as \emph{parking edges} those edges that occur in some tree along a flip sequence but are absent from both the start and final tree and conjecture that parking on the convex hull is sufficient:

\begin{conjecture}[Parking edge conjecture]
    \label{conj:parking}
    For any two plane spanning trees $T_I$ and $T_F$, there exists a shortest flip sequence from $T_I$ to $T_F$ that uses only edges of $T_I\cup T_F\cup\CH$.
\end{conjecture}

Although it was not explicitly stated as a conjecture, the above quote suggests that there always exists a shortest flip sequence in which every unhappy edge is flipped either once or twice, motivating the following conjecture.

\begin{conjecture}[Reparking conjecture]
    \label{conj:reparking}
    For any two plane spanning trees $T_I$ and~$T_F$, there exists a shortest flip sequence such that every parking edge is flipped to an edge of $T_F$. In~particular, every edge is flipped at most twice.
\end{conjecture}

This conjecture is further supported by the fact that, to date, every algorithm yielding a non-trivial upper bound on the diameter flips edges either directly to final edges, or first parks them on the convex hull, from where they are later flipped to final edges~\cite{AichholzerBBDDKLLTU24,AichholzerDV25,bjerkevik2025flipping,bousquet2025reconfiguration,bousquet2023notes}; we refer to~\cref{tab:algorithm-table} for further details.

\begin{table}[htb]
    \caption{Algorithms for upper bound constructions since 2022 and flip properties of conjectures they use. We highlight cells where algorithms use properties that are provably not always optimal.}
    \label{tab:algorithm-table}
    \centering
    \begin{tabular}{| c c | c c c |}
        \hline
        Upper bound [Reference]                      & Flip type        & \cref{conj:happy} & \cref{conj:parking}      & \cref{conj:reparking}    \\
        \hline
        $2n-\log(n)$  \cite{AichholzerBBDDKLLTU24}     & general flips    & Yes               & \cellcolor{lightgray}Yes                 & \cellcolor{lightgray}Yes                   \\
        $2n-\sqrt{n}$  \cite{bousquet2023notes}        & general flips    & No                & \cellcolor{lightgray}Yes & \cellcolor{lightgray}Yes \\
        $1.96n$      \cite{bousquet2025reconfiguration}    & general flips    & Yes               & \cellcolor{lightgray}Yes                 & \cellcolor{lightgray}Yes                   \\
        $\frac{5n}{3}-3$  \cite{bjerkevik2025flipping} & general flips    & Yes               & \cellcolor{lightgray}Yes                 & \cellcolor{lightgray}Yes                   \\
        $\frac{5n}{3}-2$  \cite{AichholzerDV25}        & compatible flips & Yes               & Yes                      & Yes                      \\
        $\frac{7}{4}(n-1)$ \cite{AichholzerDV25}     & rotations        & No                & No                       & Yes                      \\
        \hline
    \end{tabular}
\end{table}

Some preliminary results and structural insights towards the three conjectures include the following:
(i) The parking edge conjecture directly implies the happy edge conjecture~\cite{AichholzerBBDDKLLTU24}, and
(ii) the happy edge conjecture does not hold when the allowed flips are restricted to edge slides~\cite{AichholzerBBDDKLLTU24} or rotations~\cite{AichholzerDV25}; see~\cref{sec:preliminaries} for precise definitions.
In contrast, (iii) the parking edge conjecture (and consequently the happy edge conjecture) holds for compatible flips~\cite{AichholzerDV25}.
For the case of unrestricted flips, the conjectures had remained open.
The authors of~\cite{bjerkevik2025flipping}, who established the best currently known upper bound on the diameter, describe their upper bound construction as follows:
\begin{quote}
    \emph{``We note that the flip sequences we construct for \cite[Theorem 26, arXiv-Version]{bjerkevik2025flipping}, which is our strongest upper bound, have both these properties [\cref{conj:happy} and \cref{conj:parking}], providing some evidence in favor of the conjectures.''}
\end{quote}

\subsection{Our contributions}
We make progress towards the conjectures by uncovering structural insights in the different settings.
This, in turn, sheds new light on the properties of shortest flip sequences between plane spanning trees on convex point sets.
In particular, we obtain the following results.

First, we show a useful property of shortest flip sequences, namely that a flip is either final, or the removed edge would be crossed later in the sequence. 
Moreover, we show with \cref{thm:reparkingch} that the parking conjecture implies the reparking conjecture and that~\cref{conj:reparking} holds for compatible flips. Simultaneously, there always exists a shortest flip sequence in which every flip that removes a convex hull edge adds an edge of $T_F$.

\begin{restatable}[Final flip property]{theorem}{propFinalFlip}
    \label{prop:final-destination}
    For any two plane spanning trees $T_I$ and $T_F$, there exists a shortest flip sequence such that every flip is either final, or the removed edge is crossed by another edge later on.
\end{restatable}

\vspace{-2ex}

\begin{restatable}{theorem}{reparkingch}
    \label{thm:reparkingch}
    Consider any two plane spanning trees $T_I$ and $T_F$ of a point set $P$.
    \begin{enumerate}
        \item There exists a shortest flip sequence such that any convex hull edge is flipped~at~most~once.
        \item In the compatible case, all edges are flipped at most twice \textup{\textsf{(Compatible reparking property)}}.
    \end{enumerate}
\end{restatable}

Further, we refute \cref{conj:parking} by demonstrating that, for certain instances, algorithms that only introduce parking edges on the convex hull produce sequences that are far from optimal.
More precisely, we show:

\begin{restatable}{theorem}{noparking}
    \label{thm:noparking}
    There exists a family of trees $(T_{I,k})_{k\in\mathbb{N}}$ and $(T_{F,k})_{k\in\mathbb{N}}$ on $10k+2$ vertices such that the length of the shortest flip sequence that only uses parking edges from the convex hull is at least by $k$ flips longer than the shortest possible flip sequence.
\end{restatable}

Finally, we disprove \cref{conj:reparking} for the unrestricted case, obtaining the following:

\begin{restatable}{theorem}{manyreparking}
    \label{thm:manyreparking}
    There exist pairs of plane spanning trees $T_I$ and $T_F$, and $T'_I$ and $T'_F$ on $22$ and $32$ points, respectively, such that any shortest flip sequence from $T_I$ to $T_F$, respectively $T'_I$ to $T'_F$, flips the same edge $3$, respectively $4$, times.
\end{restatable}

\subparagraph*{Outline.} We begin by summarizing additional related work in \cref{sec:furtherrelated} and general definitions in~\cref{sec:preliminaries}.
We then present our positive structural results in \cref{sec:reparking-compatible-flips}, before showing that, in general, neither the parking nor the reparking property holds for unrestricted flips in~\cref{subsec:parking-edges,subsec:reparking-edges}, respectively.

\smallskip
Due to space constraints, proofs of statements marked with ($\star$) are given in the appendix.

\subsection{Further related work}\label{sec:furtherrelated}
Flip graphs of non-crossing configurations on planar point sets have been widely studied for other graph classes, including triangulations~\cite{aichholzer2015flip,eppstein2010happy,felsner2018rainbow,felsner2020rainbow,hurtado1996flipping,kanj2017computing,lubiw2015flip,pilz2014flip,POURNIN201413,SleatorTT86}, spanning paths~\cite{aichholzer2025lineartimealgorithmfinding,2023Aichholzer,akl2007planar,KleistKR24,ru-hcpgs-01}, and matchings~\cite{aichholzer2025flipping,aichholzer2025flips,binucci2025flipping,hernando2002graphs,HouleHNR05}; see also the survey by Bose and Hurtado~\cite{bose2009flipsinplanar}.
Several of the aforementioned works rely on structural observations similar to the ones we present for plane spanning trees.

In~\cite{SleatorTT86}, the authors establish the happy edge property for triangulations of convex point sets.
Moreover, they demonstrate that a greedy strategy that adds edges of a target triangulation whenever possible, always produces a shortest flip sequence.
These insights underpin several \FPT-algorithms~\cite{bosch2023improved,cleary2002restricted,li20253,lucas2010improved}, when the flip distance is a parameter.
Conversely, the happy edge property does not hold for triangulations of point sets in general position~\cite{bose2009flipsinplanar}, and counterexamples to it were key to the \NP-hardness proofs in the case of point sets in general position~\cite{lubiw2015flip,pilz2014flip}, as well as simple polygons~\cite{aichholzer2015flip}.

Similarly, the flip distance problem for perfect matchings on point sets in general position has been proven to be \NP-hard~\cite{binucci2025flipping}.
The proof likewise relies on counterexamples to the happy edge property.
When the point set is restricted to be convex, the problem is polynomial-time solvable, as an edge from the target matching can always be inserted greedily~\cite{hernando2002graphs}.

Although the happy edge property does not hold, shortest flip sequences between plane paths on convex point sets can still be computed in polynomial time~\cite{aichholzer2025lineartimealgorithmfinding}; a key ingredient of the proof is the partitioning of happy edges into ``good'' and ``bad'' happy edges, which are then treated separately.
Sometimes it is sufficient to consider only the happy edges on the convex hull;
for example, convex hull edges satisfy the happy edge property in problems such as plane spanning trees~\cite{AichholzerBBDDKLLTU24}, or odd matchings~\cite{aichholzer2025flipping} on convex point sets.
In both cases, this observation played a crucial role in deriving related \FPT-algorithms~\cite{aichholzer2025flipping,AichholzerDV25}.

Beyond complexity results, counterexamples to the happy edge property have been generalized to establish a quadratic lower bound on the sliding distance between plane spanning trees on convex point sets~\cite{aichholzer2007quadratic}.

\section{Preliminaries}
\label{sec:preliminaries}
Throughout this paper, we denote by $P\subset \mathbb{R}^2$ a set of $n$ points in convex position, that is, a~set $P$ such that no point $p\in P$ is a convex combination of any three other points.
For a pair of points $u,v\in P$, we denote by $uv$ the straight-line segment connecting the two.
If $u$ and $v$ are adjacent on the convex hull of $P$, we call $uv$ a \emph{convex hull edge} and write $uv\in \CH$.
A \emph{plane spanning tree} of $P$ is then a set of $n-1$ edges $T\subset \binom{P}{2}$ that forms a tree such that the edges' line segments are pairwise disjoint or meet at a common endpoint.
We denote the removal and insertion of an edge $e$ by $T-e$ and $T+e$, using standard set notation otherwise.

\begin{figure}[htb]%
    \centering%
    \includegraphics[page=1]{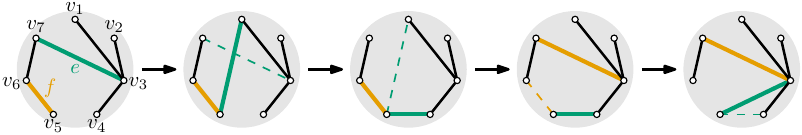}%
    \caption{%
        A flip sequence \flipSeq with ${\trace(e,\flipSeq)=(e,\:{v_1v_5},\:{v_4v_5},\:{v_3v_5})}$ and ${\trace(f,\flipSeq)=(f,\:v_3v_7)}$.
        The four flips are, in that order, crossing, a rotation, compatible, and an edge slide.
    }
    \label{fig:intro-flips}%
\end{figure}

\subparagraph*{Flip.}
A \emph{flip} in $T$ removes an edge $e\in T$ and adds an edge $f\notin T$ such that the resulting structure $T-e+f$ is again a plane spanning tree.
If the initial tree $T$ is clear from context, we denote this operation by $e\rightarrow f$ and write that \emph{$e$ is flipped to $f$}.
We differentiate between four types of flips; a flip $e\rightarrow f$ is \emph{compatible} if the straight-line segments $e$ and $f$ are either disjoint or share a common endpoint, and \emph{crossing} otherwise.
For $u,v,w\in P$, a compatible flip of the shape $uv\rightarrow uw$ is called a \emph{rotation}, and furthermore called a \emph{slide} if $vw\in T$.
Examples of all types of flips can be seen in~\cref{fig:intro-flips}.

\subparagraph*{Flip graph.}
The \emph{flip graph} of plane spanning trees of $P$ is the graph that has all such trees as its vertices.
Two spanning trees share an edge in the flip graph if and only if they can be transformed into one another via a single flip.
Given a pair of trees $T_I$ and $T_F$, a~\emph{flip sequence}~$\flipSeq$ from~$T_I$ to $T_F$ is then a path in the flip graph denoted by
a sequence of trees~${T_I=T_0,T_1,\ldots,T_{k-1},T_k=T_F}$ such that consecutive trees only differ by a single flip; its length is the number $k$ of flips in the sequence.
The \emph{flip distance} of $T_I$ and $T_F$, denoted by $d(T_I,T_F)$, is the length of a shortest sequence.
Furthermore, the edges $T_I\cap T_F$ are \emph{happy}, with all other edges being \emph{unhappy}.
If a flip introduces an edge~$f\notin T_F$, it is commonly referred to as a \emph{parking flip}, and $f$ as its \emph{parking edge}.
The~\emph{diameter} of a flip graph denotes the maximum flip distance between any pair of trees.

\subparagraph*{Trace.}
We introduce additional notation to analyze flip sequences.
Namely, we introduce structures that allow us to not only represent changes within trees under flip operations, but also how these flip operations affect individual edges.
Given a flip sequence $\flipSeq=T_0,\ldots, T_k$, we define the \emph{trace} of an edge $e\in T_0$ as the ordered sequence of positions that it is flipped to; an example is shown in~\cref{fig:intro-flips}.
Formally, this corresponds to a tuple of edges:
\[
    \trace(e,\flipSeq)=\trace(e, T_0,\ldots,T_k) =
    \left\{\begin{array}{@{}l@{}l@{\quad}l}%
                (e) &&\text{if $k=0$,} \\ 
               &\trace(e,T_1,\ldots, T_k)&\text{if $e\in T_0\cap T_1$, or}\\
               (e)\parallel\:&\trace(f,T_1,\ldots,T_k)&\text{if $\{e,f\}=T_0\Delta T_1$.}
    \end{array}\right.
\]
We say that the \emph{length of a trace} is the number of flips in it, which corresponds to the number of elements minus one: The trace $(e_1,\ldots, e_{\ell})$ corresponds to the flips $e_i\to e_{i+1}$ for $1\leq i< \ell$.
Observe that the length of $\flipSeq$ is exactly the sum of all trace lengths. 
A flip~$e\to f$ in $\flipSeq$ is \emph{final} if the trace of $f$ on the remaining sequence is $(f)$, and \emph{non-final} otherwise.
A~happy edge $e\in T_0\cap T_k$ is \emph{fixed} in the flip sequence~$\flipSeq$ if it has trace length zero, or equivalently $\trace(e,\flipSeq)=(e)$.
This induces two properties of happy edges based on shortest flip sequences. 
A happy edge is \emph{weakly fixed} if it is fixed in some shortest flip sequence, and \emph{strongly fixed} if this applies to all such~sequences.
\medskip

We note that~\cref{conj:happy,conj:parking,conj:reparking} can also be stated in terms of trace properties:
The happy edge conjecture (\cref{conj:happy}) states that \emph{all happy edges are weakly fixed}.
Similarly, the parking conjecture (\cref{conj:parking}) states that there always exists a shortest flip sequence in which \emph{all but the first and final element of all traces are in \CH}, and the re-parking conjecture~(\cref{conj:reparking}) states 
the existence of a shortest flip sequence in which each edge of $T_I$ has trace length at most two.
\section{Positive results on structural properties}
\label{sec:reparking-compatible-flips}

We generalize a property shown by Aichholzer, Dorfer, and Vogtenhuber~\cite[Lemma 20]{AichholzerDV25}.
They showed that in every compatible flip sequence any parking flip that parks on a diagonal can be replaced with a parking flip that parks on the convex hull.
Actually, it is not required that the entire flip sequence is compatible, but only either the flip that inserts the parking edge or the flip that removes the parking edge has to be compatible.
This already follows from their proof.
For the sake of self-containment, we include the adapted proof and rewrite the statement in terms of our definitions.

\begin{restatable}[$\star$]{lemma}{lemmaPreHappy}
    \label{lem:prehappy3}
    Let there be a flip sequence between two trees $T_I$ and $T_F$ and an edge~$e$ such that $trace(e)=(e,f,e')$, where $f$ is a diagonal and either the flip $e\to f$ or the flip~${f \to e'}$ is compatible.
    Then there exists a flip sequence from $T_I$ to $T_F$ of the same length in which $trace(e)=(e,h,e')$ such that $h$ is a convex hull edge. The trace of any other edge remains~unchanged.
\end{restatable}

In order to study the reparking property, we start by observing sufficient properties for edges to have trace length $1$ or, equivalently, necessary conditions for an edge to have larger trace length.

\propFinalFlip*
\begin{proof}

    \begin{figure}[htb]%
        \captionsetup[subfigure]{justification=centering}%
        \begin{subfigure}[t]{\columnwidth/3}%
            \centering%
            \includegraphics[page=1]{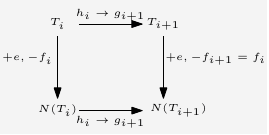}%
            \subcaption{}
            \label{fig:same}
        \end{subfigure}%
        \begin{subfigure}[t]{\columnwidth/3}%
            \centering%
            \includegraphics[page=2]{figures/Destination_final}%
            \subcaption{}
            \label{fig:different}
        \end{subfigure}%
        \begin{subfigure}[t]{\columnwidth/3}%
            \centering%
            \includegraphics[page=3]{figures/Destination_final}%
            \subcaption{}%
            \label{fig:last}
        \end{subfigure}%
        \caption{Commutative diagrams on how flips may differ. Horizontal arrows represent flips and are therefore labeled in flip notation. Vertical lines represent the normalization operation and are denoted in set notation since they are not flips in the traditional sense. From left to right: (a) the edges $h$ and $g$ do not lie on the same cycle as $e$; (b) the two edges $f$ and $g$ lie on the same cycle as $e$; and (c) the last flip on the cycle that contains $e$.}
        \label{fig:final}
    \end{figure}

    For a contradiction, let $T_I$ and $T_F$ be a pair of trees where (1) any shortest flip sequence between them contains an edge $e$ with trace length more than 1 such that $e$ is not crossed by any other edge of the flip sequence, and (2) the flip distance between $T_I$ and $T_F$ is minimal over all pairs of trees with such shortest flip sequences.
    Let $T_I=T_0$, $T_1,\ldots, T_k = T_F$ be such a flip sequence from $T_I$ to $T_F$.
    We can assume without loss of generality that the trace of $e$ is $(e,f,e')$ and $e \in T_I = T_0$ and $e' \in T_F = T_k$, as otherwise we could consider a closer pair and shorten the flip sequence.
    In other words, the flip $e\to f$ happens from $T_I=T_0$ to $T_1$ and the flip $f\to e'$ is the last flip, namely from $T_{k-1}$ to $T_k = T_F$.
    
    For our purposes we apply a technique called \emph{normalization} as introduced for triangulations in~\cite[Lemma~3]{SleatorTT86}.
    This concept has since been modified for matchings~\cite{hernando2002graphs} and trees~\cite[Proposition~18]{AichholzerBBDDKLLTU24}, which we further adapt.
    For trees, a normalization is applied to a flip sequence and defined with respect to an edge $e$ that is not crossed by any other edge that is part of this flip sequence.
    A normalization maps a tree $T$ to another tree $N(T)$ such that the mapping $N(\cdot)$ only depends on the original flip sequence and $e$.
    For the original flip sequence $\flipSeq=T_0$, $T_1,\ldots, T_k$ and an edge $e$ we want the normalization $\flipSeq_N = N(T_0)$, $N(T_1),\ldots,N(T_{k-1})$, $N(T_k)$ to have the following properties:
    \begin{enumerate}[\quad (1)]
        \item $N(T_0)=T_0$ and $N(T_k) = T_k$.
        \item For two consecutive trees $T_i$, $T_{i+1}$, we have that $N(T_i)$ and $N(T_{i+1})$ either coincide or only differ in a single flip.
        \item The $\trace(e,\flipSeq_N)$ has length $1$.
    \end{enumerate}
    Note that in Property (2) for our sequence $\flipSeq$ it cannot happen that two consecutive normalized trees coincide, as this would give us a shorter flip sequence and thus a contradiction.

    We now explain how we obtain this normalization for a tree $T_i$ of a flip sequence w.r.t.\ $e$.
    If $e$ is already part of $T_i$, then we simply set $N(T_i) = T_i$.
    Otherwise, we add $e$ to $T_i$.
    Then~$T_i$ contains a cycle $c_i$ with $e$ on it.
    Define $N(T_i)=(T_i\cup \{e\})\setminus\{f_i\}$ where~$f_i$ is either the first edge on $c_i$ that is removed during the remaining flip sequence from $T_i$ to~$T_k$ or, if no such edge exists, set $f_i=e$.

    We now argue that $N(\cdot)$ has all the desired properties.
    For Property (1), we observe that $N(T_0)$ contains $e$ and by construction $N(T_0)=T_0$.
    If $e$ is in $T_k$, by the same argument it follows that $N(T_k) = T_k$.
    If $e$ is not in $T_k$, we obtain the described plane cycle $c_k$ by adding $e$.
    Since the flip sequence ends at $T_k$, no edge of $c_k$ gets removed in a later point of the flip sequence, and we remove $e$ again to obtain $N(T_k)$.
    We conclude that in both cases~$N(T_k)= T_k$.

    For Property (2), in the flip from $T_i$ to $T_{i+1}$ we denote the removed edge by $h_i$ and the added edge by $g_{i+1}$.
    Consider the cycle $c_i$ containing $e$ in $T_i\cup \{e\}$.
    If $h_i$ does not belong to~$c_i$, then $c_{i+1} = c_{i}$ and $f_{i+1} = f_{i}$.
    The normalized trees $N(T_i)$ and $N(T_{i+1})$ also differ in the flip $h_i\to g_{i+1}$.
    This case is depicted via the commutative diagram in \cref{fig:same}.
    If $h_i$ is part of $c_i$ then in the normalization we have $h_i=f_i$.
    Moreover, $c_i \neq c_{i+1}$ as we remove~$h_i$ from $T_i$ and add $g_{i+1}$ to obtain $T_{i+1}$.
    When removing $f_{i+1}$ from $N(T_i)$ and adding~$g_{i+1}$ we obtain $N(T_{i+1})$.
    In other words, the flip $h_i\to g_{i+1}$ from $T_i$ to $T_{i+1}$ is transformed to the flip $f_{i+1} \to g_{i+1}$ from $N(T_i)$ to $N(T_{i+1})$, see~\cref{fig:different}.
    Thus we conclude that Property~(2)~holds.

    For Property (3), observe that as long as in step $k$ there is an edge $f_k \neq e$ in the cycle $c_k$ that is later on flipped, $e$ is part of the tree $N(T_k)$.
    Consider the last flip in $\flipSeq_N$ for which an edge is added to $c_k$ for some $k$.
    Then in this flip $e$ was removed from $N(T_k)$ and flipped to an edge in the cycle $c_k$.
    As now the cycle $c_k$ consists entirely of edges which are not flipped anymore in the remainder of $\flipSeq$, Property~(3) follows.
    
    With the new flip sequence $N(T_0)$, $N(T_1),\ldots,N(T_{k-1})$, $N(T_k)$ we have made sure that~$e$ is no longer an edge that is flipped multiple times, but not crossed at any point of the flip sequence.
    However, we might have introduced another edge with these properties in the new flip sequence.
    If no such edge exists then we are done.
    Thus for the remainder of the proof we assume that yet another edge is flipped multiple times but not crossed by any other edge of the flip sequence $N(T_0)$, $N(T_1),\ldots,N(T_{k-1})$, $N(T_k)$.
    In order to be able to show that~\cref{prop:final-destination} holds for all edges at the same time we, resort to a technique used in~\cite[Lemma~20]{AichholzerDV25}.
    
    If an edge $s=ab$ is contained in all trees of a flip sequence, then $s$ splits the convex point set into two sides $A$ and $B$, the vertices that appear in clockwise order between $a$ and $b$ and the vertices that appear in counterclockwise order.
    The vertices $a$ and $b$ count towards both, side $A$ and side $B$.
    For every tree, we obtain two subtrees, the induced subtrees on $A$ and~$B$.
    Every flip in one of the trees is completely independent of the flips in the other side.
    Thus, we can interleave both ordered subsequences that flip edges in only one of the two sides to obtain a valid flip sequence with the same initial and final trees. 
    We will use this to argue that the only flip in $\trace(e,\flipSeq_N)$, denoted by $e\to g$, is independent of all subsequent flips in~$\flipSeq_N$.
    This will then imply that $e\to g$ can be performed as the last flip.

    To this end, assume that the flip from $N(T_{i})$ to $N(T_{i+1})$ is the final flip that removes~$e$ and inserts $g$.
    We split the convex point set along all edges of $c_{i+1}\setminus\{e,g\}$. Every tree $N(T_i), N(T_{i+1}), \ldots, N(T_k)$ now splits into $\lvert c_{i+1} \rvert -1$ trees by splitting along the edges of $c_{i+1}\setminus\{e,g\}$ (recall that $c_j=c_{i+1}$ for any $j \geq i+1$). Flip sequences of these trees in the resulting subsets are independent of another and can be interleaved in any order.
 
    After splitting the point set, there exists a particular subset $C$, as depicted in \cref{fig:baum}, of vertices that contains the vertices of both, $e$ and $g$.
    By the previous discussion, we can perform the flip sequence in $C$, that also performs the flip $e\to g$, as the last flip sequence among all flip sequences in the subsets.

    It remains to argue that we can perform the flip $e\to g$ as the last flip in~$C$. For this, we argue that all other flips can be performed, independently of whether $e$ or $g$ is present in the current tree. $C$ has three parts for which $e$ and $g$ form the border: The side $C_e$ of~$e$ that does not contain $g$, the side $C_g$ of $g$ that does not contain $e$, and the part between $e$ and~$g$. The latter part contains exactly the edges of $c_{i+1}$. Therefore, no flips other than $e\to g$ happen in this part, since no edge of $c_{i+1}$ is flipped.

     \begin{figure}[htb]%
        \captionsetup[subfigure]{justification=centering}%
        \begin{subfigure}[t]{\columnwidth/3}%
            \centering%
            \includegraphics[page=1]{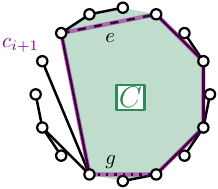}%
            \subcaption{%
            Illustration of the cycle $c_{i+1}$.
            }
            \label{fig:baum}%
        \end{subfigure}%
        \begin{subfigure}[t]{0.65\columnwidth}%
            \centering%
            \includegraphics[page=2]{tree_cycles}%
            \subcaption{Changing $c_{h',g'}$ and $\hat{c}_{h',g'}$ into one another}
            \label{fig:baum2}
        \end{subfigure}%
        \caption{Illustration of the concepts of the proof of \cref{prop:final-destination}.}
        \label{fig:baum3}
    \end{figure}

    Consider a flip in $C_e$. If we want to flip $h'\to g'$ in $C_e$ then, when adding $g'$ to the involved tree in $C_e$, we obtain a cycle $c_{h',g'}$ that contains $h'$ and $g'$. 
    If $c_{h',g'}$ contains neither $e$ nor~$g$, then obviously it does not matter if we perform the flip $h'\to g'$ before of after $e\to g$. 
    If one of $e$ or $g$ is in $c_{h',g'}$ (note that it is not possible that both are in the cycle at the same time), then we argue that they can be replaced by each other. The alternative cycle containing $g'$ and $h'$ is then the symmetric difference of $c_{h',g'}$ and $c_{i+1}$, namely $\hat{c}_{h',g'} = c_{h',g'} \Delta c_{i+1}$. See \cref{fig:baum2} for an example. If $c_{h',g'}$ contains the edge $e$, then $\hat{c}_{h',g'}$ contains the detour via $c_{i+1}\setminus\{e\}$, including $g$. Vice, versa if $c_{h',g'}$ contains $g$, and thus $c_{i+1}\setminus\{e\}$, then $\hat{c}_{h',g'}$ contains $e$.
    For flips in $C_g$ we argue symmetrically.

    To finally prove the theorem, we conclude that we can reorder the flips in the flip sequence such that the flip $e\to g$, which is a final flip, is the last flip of the flip sequence. As $e$ was an edge of the initial tree $T_I$ and has now trace length 1, we can remove the flip $e\to g$ from our sequence without changing whether the conditions of \cref{prop:final-destination} are fulfilled. Therefore, by omitting the flip $e\to g$ and adjusting the target tree accordingly to $T_k + e - g$, we have produced a flip sequence that is one flip shorter than our initial sequence. Recall that we assumed that an edge that is flipped multiple times but not crossed by any other edge was introduced to the sequence $N(T_0)$, $N(T_1),\ldots,N(T_{k-1})$, $N(T_k)$. Thus our new, shorter, flip sequence also contains that edge. This provides a contradiction to our assumption that the flip sequence from the start was the one with the minimum number of flips among all sequences that do not fulfill the statement of the theorem.
\end{proof}

We continue by showing that~\cref{prop:final-destination} is strong enough to settle the reparking property for compatible flip sequences, as well as for convex hull edges of $T_I$ in the general~case.

\reparkingch*

\begin{proof}
    The first part follows from the fact that convex hull edges are never crossed together with~\cref{prop:final-destination}. 

    For the second part we apply that the parking property holds for compatible flips~\cite[Lemma 20]{AichholzerDV25}. Since there exists a shortest flip sequence that only introduces parking edges on the convex hull, we can apply part~(1) of this theorem to the parking edges on the convex hull. Every edge is then flipped at most twice: once to park on the convex hull and once for the final flip.
\end{proof}

\begin{corollary}
    For every pair of trees $T_I$, $T_F$, there exists a shortest flip sequence from $T_I$ to $T_F$ in which every flip that removes a diagonal~$e$ adds one of the three following:
    \begin{itemize}
        \item An edge of $T_F$ via a final flip.
        \item An edge on the convex hull.
        \item A parking edge that crosses $e$.
    \end{itemize}
    Further, any flip that removes a convex hull edge is final.
\end{corollary}

\begin{proof}
    The only case left is when a diagonal $e$ is replaced by a non-crossing diagonal. 
    By \cref{lem:prehappy3}, there exists an alternative flip sequence that replaces $e$ with a convex hull~edge.
\end{proof}

\section{Counterexamples for general flips}
\label{sec:counterexamples}

In the previous section, we identified structural conditions that any counterexample to the parking and reparking conjectures must satisfy for general flips.
With these conditions in hand, we now construct specific pairs of trees, $T_I$ and $T_F$, that not only meet these requirements but also go further, yielding shortest flip sequences that violate the conjectures.

Through a combination of case distinctions, additional observations on flip sequences, and support from a computer program, we verify that the proposed flip sequences are indeed the shortest possible flip sequences.
Moreover, we demonstrate that no alternative shortest flip sequence can satisfy the conjectures.

\subsection{The parking edge property}
\label{subsec:parking-edges}

In this section, we prove \cref{thm:noparking}, thereby disproving \cref{conj:parking}. The theorem consists of two parts, which we state and prove separately for convenience.

\begin{figure}[htb]%
    \centering%
    \includegraphics[page=1]{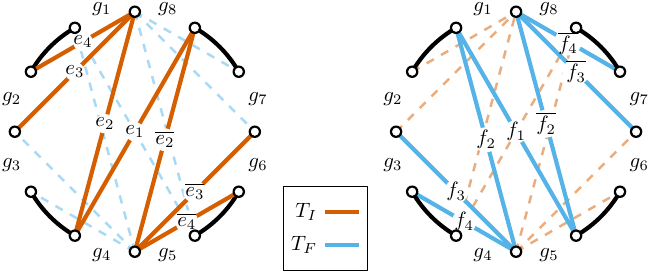}%
    \caption{A counterexample to the parking edge conjecture using $n=12$ points.}
    \label{fig:parking-counterexample}
\end{figure}

\begin{restatable}[$\star$]{lemma}{lemmaParkingEdgeProperty}
    \label{thm:fallacy}
    The parking edge conjecture does not hold for general flips and $n\geq 12$ points in convex position.
    Moreover, a minimum counterexample need not be restricted to incompatible flips, and may even contain edge~slides in all shortest flip sequences.
\end{restatable}
\begin{proof}[Proof sketch.]
    Consider the pair of trees visualized in~\cref{fig:parking-counterexample}.
    Our aim is to transform $T_I$ into $T_F$.
    In doing so, we demonstrate that all shortest flip sequences must necessarily park an edge on a diagonal.
    In~\cref{fig:parking-counterexample-opt}, we depict a  flip sequence that uses a diagonal parking edge~$p_1$ (green).
    Moreover, it requires exactly one parking flip, $e_1 \to p_1$, while all other flips are final.
    Furthermore, since every edge in the symmetric difference~$T_I\Delta T_F$ is crossed at least twice, the first flip of any sequence from~$T_I$ to~$T_F$ must necessarily be a parking flip. 
    Therefore, the sequence illustrated in \cref{fig:parking-counterexample-opt} is indeed a shortest flip sequence.
    We note that this flip sequence does also contain two edge slides; these are highlighted in purple.

    It remains to show that any parking flip that parks to the convex hull instead of~$p_1$ requires at least one additional flip.
    To this end, we show that no flip sequence beginning with a parking flip to the convex hull can reach~$T_F$ exclusively using final flips thereafter.
    In~\crefrange{fig:parking-counterexample-park-1}{fig:parking-counterexample-park-4}, we indicate all possible first flips (up to symmetry) that replace a diagonal of~$T_I\setminus T_F$ with an edge on the convex hull.
    For each such flip, \cref{fig:cases} depicts a flip sequence that continues performing perfect flips until it becomes impossible to add any further edge of the final tree without making another parking flip.
    We emphasize that these sequences are not unique, as sometimes multiple edges can be flipped to the same edge, or flips can be performed in a different order yet still produce the same result.
    For instance, \cref{fig:example-parking-to-gap-8} shows the graph of all possible sequences that start with~${e_2\to g_8}$ and otherwise consist of final flips.
    The same reasoning applies to the other cases in which a diagonal is parked on the convex hull, as each can be analyzed in a similar manner.
    \begin{figure}[ph]
        \begin{subfigure}[t]{\columnwidth}%
            \centering%
            \includegraphics[page=2]{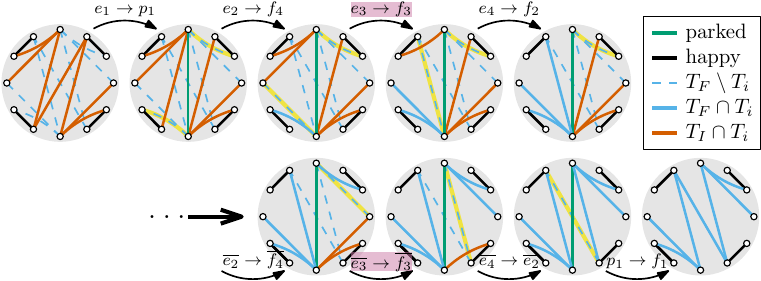}
            \subcaption{A shortest flip sequence for the instance shown in~\cref{fig:parking-counterexample}. The two highlighted flips are slides.}%
            \label{fig:parking-counterexample-opt}%
        \end{subfigure}%
        \par\medskip%
        \begin{subfigure}[t]{\columnwidth}%
            \centering%
            \includegraphics[page=3]{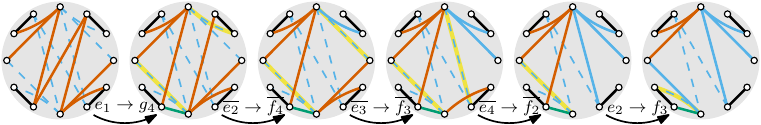}
            \subcaption{A ``maximal'' flip sequence starting with~$e_1\to g_4$. Note that~$e_1\to g_8$ is symmetric.}%
            \label{fig:parking-counterexample-park-1}%
        \end{subfigure}%
        \par\medskip%
        \begin{subfigure}[t]{\columnwidth}%
            \centering%
            \includegraphics[page=4]{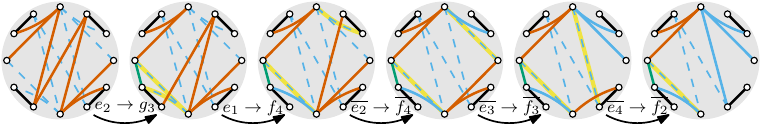}
            \subcaption{A ``maximal'' flip sequence starting with~$e_2\to g_3$. Note that~$\overline{e_2}\to g_7$ is symmetric.}%
            \label{fig:parking-counterexample-park-2}%
        \end{subfigure}%
        \par\medskip%
        \begin{subfigure}[t]{\columnwidth}%
            \centering%
            \includegraphics[page=5]{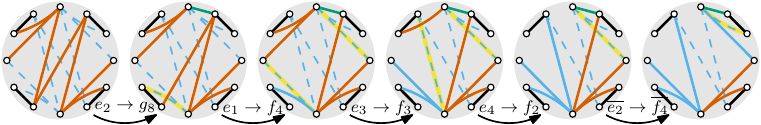}
            \subcaption{A ``maximal'' flip sequence starting with~$e_2\to g_8$. Note that~$\overline{e_2}\to g_4$ is symmetric.}%
            \label{fig:parking-counterexample-park-2b}%
        \end{subfigure}%
        \par\medskip%
        \begin{subfigure}[t]{\columnwidth}%
            \includegraphics[page=6]{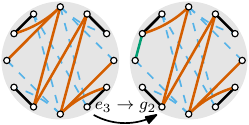}%
            \hfill%
            \includegraphics[page=7]{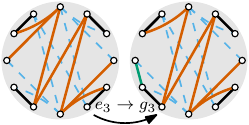}%
            \hfill%
            \includegraphics[page=8]{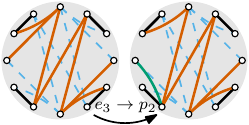}%
            \par\bigskip%
            \includegraphics[page=9]{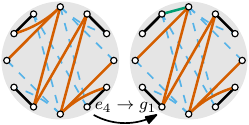}%
            \hfill%
            \includegraphics[page=10]{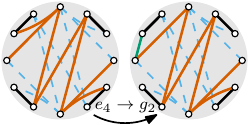}%
            \hfill%
            \includegraphics[page=11]{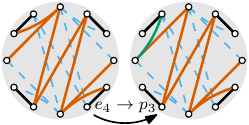}%
            \subcaption{Parking either $e_3$ or $e_4$ leads to dead ends after a single flip. Note that $\overline{e_3}$ and $\overline{e_4}$ have symmetric flips.} %
            \label{fig:parking-counterexample-park-4}%
        \end{subfigure}%
        \caption{We illustrate a shortest flip sequence (a) that parks the edge~$1$ diagonally, alongside the maximal number of perfect flips possible if another edge pair is chosen for parking in the first flip.
        For clarity edges of~$T_F$ that are can currently be inserted are emphasized.}
        \label{fig:cases}
    \end{figure}
\end{proof}

\begin{figure}[phtb]%
    \centering%
    \includegraphics{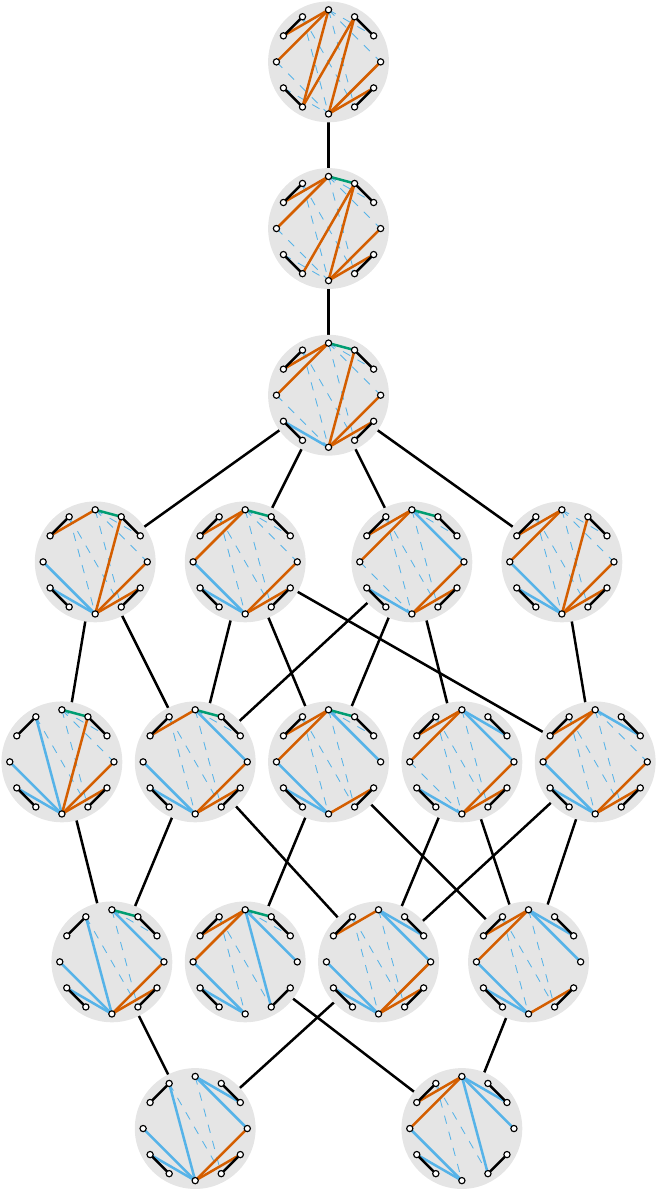}%
    \caption{The subgraph of reachable trees after~$e_2\to g_8$ without performing a second parking flip.}
    \label{fig:example-parking-to-gap-8}
\end{figure}

To compute all relevant flip sequences for the example shown in~\cref{fig:parking-counterexample}, we employ a straightforward computer program~\cite{visual-and-verifier}, which systematically generates and tracks all possible sequences.
The program initializes a set of edges with the edges of the initial tree~$T_I$ and then simply works in a depth first search (DFS) like manner. For every given set of edges it considers all possible legal flips (such that after the flip the set of edges still forms a spanning tree). For each such flip, it performs the flip and then increases a counter in case that this flip was not perfect, that is, did not insert an edge of the final tree. Only if this counter does not exceed the maximum allowed number of non-perfect flips (one in our case) the program continues recursively. Otherwise the flip is undone and the next legal flip is performed, etc.
Whenever all edges of the set are from the final tree~$T_F$ another counter for the number of shortest flip sequences is increased. The program terminates once the DFS is completed.

For the trees shown in~\cref{fig:parking-counterexample} the output of the program gives~$20$ different flip sequences of length 8. That these are indeed precisely all the sequences which contain the parking flip~${e_1\to p_1}$ can be verified by counting all shortest flip sequences depicted in \cref{fig:counterexample-dag}.
Observe that there is some left-right symmetry once the flip~${e_1\to p_1}$ has been made. Then the flips~${e_2\to f_4}$, ${e_3\to f_3}$, and~${e_4\to f_2}$ have to be performed in that order, but can interleave with the symmetric flips
${\overline{e_2}\to \overline{f_4}}$, ${\overline{e_3}\to \overline{f_3}}$, and~${\overline{e_4}\to \overline{f_2}}$ in an arbitrary order. As this is the only place where the order is flexible, we get~${6 \choose 3}=20$ different flip sequences of length 8, precisely the number reported by the program.

We used a modified version of our program to generate the relevant flip graph, that is, all trees and flips between them that can be reached from the initial tree~$T_I$ with a flip sequence that includes at most one non-perfect flip.
This graph has 118 different trees as nodes, and 214 edges (flips) between them.
The full graph is depicted in~\cref{fig:counterexample-dag} in Appendix~\ref{app:figures}.

We~further reduced these numbers by forcing the first (non-perfect) flip to be a flip onto the convex hull.
This is possible because to verify the correctness of our statement, we only need to consider shortest flip sequences that do not have a parking flip to a diagonal (recall that the first flip has to be a parking flip).
This results in a graph with only 67 trees and 120 edges which nicely shows all relevant cases, including the one elaborated in detail in the proof of~\cref{thm:fallacy}.
Both variants of the flip graph can be interactively explored~\cite{visual-and-verifier}.

For the reader's convenience, all relevant components of the graph can be found in Appendix~\ref{app:figures} alongside~\cref{fig:counterexample-dag}; each possible parking flip leads to a unique (up to symmetry) portion of the subgraph that remains reachable without making a second non-perfect flip.
These can be found in~\cref{fig:example-parking-to-gap-8,fig:descends_156_7,fig:descends_383,fig:descends_80}.

\smallskip
Restricting flips to parking edges on the convex hull incurs a single additional flip, as we have seen in~\cref{fig:parking-counterexample}.
We next strengthen this result, proving the existence of instances that require~$\Omega(n)$ additional flips under the same constraint.

\begin{lemma}
    For any $k\in \mathbb{Z}^+$, there exist two non-crossing spanning trees on $n\geq10k+2$ points such that a shortest flip sequence takes $8k$ flips, whereas a shortest flip sequence that parks only on convex hull edges takes at least $9k$ flips.
\end{lemma}
\begin{figure}[htb]%
    \centering%
    \includegraphics[page=12]{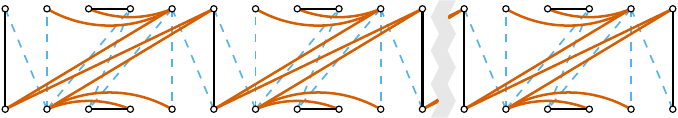}%
    \caption{A counterexample to the parking edge conjecture using $n=10k+2$ points.}
    \label{fig:parking-counterexample-linear-gap}
\end{figure}

\begin{proof}
    We glue $k$ instances of trees $T_I$ as seen in~\cref{fig:parking-counterexample} along happy edges. Our target structure consists of $k$ instances of $T_F$ glued along the same happy edges. If no happy edges are flipped in a flip sequence this is equivalent to having $k$ separate flip sequences in every individual occurrence of $T_I$ and we need $8$ flips in an optimal flip sequence compared to $9$ if we only use parking edges on the convex hull. Summing over all occurrences of $T_I$ gives an optimal flip sequence of length $8k$ and a minimum flip distance when only using parking edges on the convex hull of $9k$.

    Next, assume for a minimum counterexample that all $k-1$ consecutive happy edges are removed at some point during the flip sequence and added back in a later flip. We argue that this flip sequence has to take at least $10k-2$ flips, unless it introduces parking edges that are not on the convex hull. Assume the opposite. Now, the flip sequence removes and adds happy edges and since the flip sequence does not introduce diagonal parking edges, the happy edges are never crossed by any edge during the flip sequence. Proposition 18 in~\cite{AichholzerBBDDKLLTU24} states that in any minimum flip sequence, for any happy edge that is removed there later has to exist some edge that would have crossed this happy edge. If such a crossing edge does not exist, the proof of the proposition gives a construction of a flip sequence that is at least one flip shorter and that does not remove the happy edge. By applying this result, we can iteratively construct flip sequences that do not remove and add one happy edge such that the resulting flip sequence is by one flip shorter than the original flip sequence. Applying the result iteratively for all happy edges, gives a flip sequence of length at most $9k-1$ that does not flip any happy edges and still does not use any parking edges outside the convex hull. By the pigeonhole principle, there exists one occurrence of $T_I$ that is flipped to an occurrence of~$T_F$ in only $8$ flips without introducing diagonal parking edges, contradicting~\cref{thm:fallacy}.
\end{proof}

On the positive side, we demonstrate that flip sequences that park solely on the convex hull remain close to an optimal flip sequence.
More precisely, we prove that a diagonal parking flip can be replaced by two convex hull parking flips. Applying the following proposition once for every diagonal parking flip yields a flip sequence without any parking flip, but at the cost of one additional flip per iteration.

\begin{proposition}
    \label{prop:twopark}
    A flip sequence that adds and removes a parking diagonal $d$ by non-compatible flips can be replaced with a flip sequence that does not perform a diagonal parking flip that involves $d$ and that is only one flip longer. In particular, the resulting flip sequence performs strictly fewer diagonal parking flips than the original one.
\end{proposition}

\begin{proof}
    Without loss of generality, assume that in the flip from $T_I = T_0$ to $T_1$ an edge $e$ is flipped to a parking edge $f=(v_a,v_b)$ via a non-compatible flip and in the flip from $T_{k-1}$ to $T_k = T_F$ $f$ is removed by a non-compatible flip and exchanged for an edge $e'$. 
    Then, $f$ separates the point set into two sides $A$ and $B$, the points that appear in clockwise order from $v_a$ to $v_b$, and the points in counterclockwise order. 
    We include $v_a$ and $v_b$ in both sets. 
    In the flip sequence from $T_1$ to $T_{k-1}$ flips that exchange edges in $A$ or $B$ can be executed independently from the other side. If we remove the edge $f$ from $T_1$, we obtain two subtrees and there exists a convex hull edge between two consecutive convex hull vertices $v_i,v_{i+1}$ in $A$ that connects the two subtrees. Without loss of generality $v_i$ lies in the same subtree as $v_a$.

    We flip from $T_0$ to $(T_1\setminus\{f\})\cup\{h\}$. Then any flip in the flip sequence from $T_1$ to $T_{k-1}$ that happens in side $B$ can still be executed. Any flip that added and removed an edge that bounded a cycle that contains $f$ now contains a path that walks first from $v_a$ to $v_i$ to the subtree in $A$, then traverses $h$ and later moves from $v_{i+1}$ to $v_b$.

    After all flips in $B$ have been executed and the side $B$  coincides with the side~$B$ in~$T_{k-1}$, we perform a flip that removes $h$ and adds a convex hull edge $h'$ that joins the two resulting subtrees in $B$. We continue to execute all flips in side~$A$. Again, any flip that involves a cycle that contains $f$ now contains a path through side $B$ and the convex hull edge $h'$. After performing all flips in $A$, we reached the tree~${T_{k-1}\setminus\{f\}\cup\{h\}}$. As a last flip, we remove $h$ and add $e'$ to obtain the tree $T_k$. We made all the flips that happened in the original flip sequence. Only the trace of edge $e$ has been changed from~$(e,f,e')$ to $(e,h,h',e')$. Therefore, the resulting flip sequence is exactly one flip longer than the original flip sequence.
\end{proof}

We note that the flip sequence derived due to the proof of~\cref{prop:twopark} is not ``nice'', in the sense that it includes a non-final flip from one convex hull edge to another. 
Using~\cref{prop:final-destination}, we can reorder certain flips, producing a ``nicer'' flip sequence.

\subsection{The reparking property}
\label{subsec:reparking-edges}

From~\cref{prop:final-destination,lem:prehappy3}, we derive necessary conditions for a counterexample to the reparking property as follows:

\begin{proposition}
    \label{prop:reparkingsetup}
    Suppose there exists a shortest flip sequence with an edge $e$ with $trace(e)=(e,e_1,e_2,e_3)$ of length $3$.
    If no other shortest flip sequence exists in which $e$ has shorter trace, the following conditions must hold:
    \begin{description}
        \item[\emph{(1)}] $e$ and $e_3$ are diagonals.
        \item[\emph{(2)}] $e$ and $e_1$ cross.
        \item[\emph{(3)}] $e_1$ and $e_2$ cross.
        \item[\emph{(4)}] $e_2$ and $e_3$ cross.
    \end{description}
\end{proposition}

\begin{proof}
    We approach the different conditions one by one.
    \begin{description}
        \item[(1)] If $e$ is a convex hull edge, it is in particular not intersected by any other edge. By~\cref{prop:final-destination}, we conclude that there exists a shortest flip sequence that flips $e$ to its final position. 
        If $e_3$ is a convex hull edge, we apply~\cref{prop:final-destination} to the reversed flip sequence.
        \item[(2)] If $e$ and $e_1$ do not cross, then $e_1$ is added by a compatible flip. We can apply~\cref{lem:prehappy3} to obtain a shortest flip sequence in which $e_1$ is replaced by a convex hull edge. In a second step, we apply~\cref{prop:final-destination} to obtain a shortest flip sequence in which $e_1$ is flipped to its final position. Therefore, we have constructed a shortest flip sequence in which $e$ has trace length $2$.
        \item[(3)] If $e_1$ and $e_2$ do not cross, we first apply~\cref{lem:prehappy3} to obtain a shortest flip sequence in which $e_2$ is replaced with a convex hull edge $h_2$. Now we can reverse the flip sequence. In the new flip sequence $e_3$ is flipped to a convex hull edge. Then, we apply~\cref{prop:final-destination} to obtain a shortest flip sequence in which $h_2$ is flipped to its final position in the reversed flip sequence, that is, $e$. By reversing the flip sequence again, we obtain a shortest flip sequence in which $trace(e)=(e,h_2,e_3)$, in particular, $e$ has trace length $2$.
        \item[(4)] If $e_2$ and $e_3$ do not cross, we reverse the flip sequence and proceed as in Case (2).\qedhere
    \end{description}
\end{proof}

\begin{figure}[htb]
    \centering
    \includegraphics[page=2]{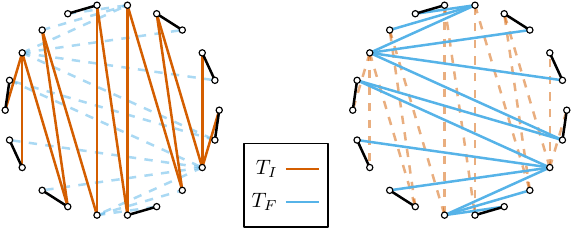}
    \caption{A counterexample to the reparking property.}
    \label{fig:reparking_counterexample}
\end{figure}

\cref{fig:reparking_counterexample} provides a counterexample towards the reparking property. This counterexample was constructed with the conditions of~\cref{prop:reparkingsetup} in mind. Any shortest flip sequence from $T_I$ in red to $T_F$ in blue contains an edge with trace length $3$. A sketch of a shortest flip sequence is given in~\cref{fig:reparking_counterexample_sequence}. First, the central diagonal in red is flipped to the green diagonal. Then, we can perform six perfect flips that insert edges from $T_F$. Afterwards, the green diagonal is flipped again. We can perform six additional perfect flips, before we can flip the green edge to its final position.

\begin{figure}
    \centering
    \includegraphics[page=3]{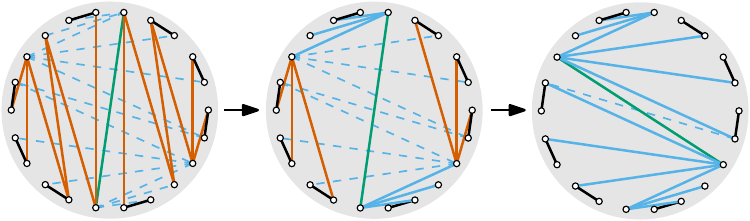}
    \caption{A counterexample to the reparking property.}
    \label{fig:reparking_counterexample_sequence}
\end{figure}

The flip sequence sketched in~\cref{fig:reparking_counterexample_sequence} is the unique shortest flip sequence (up to ordering) of the six perfect flips that occur between consecutive flips that involve the green edge. 
To~verify this, we again used our program~\cite{visual-and-verifier}. Here the relevant part of the flip graph has 5086 vertices (which are at most two non-perfect flips away from~$T_I$). 
The size of this graph seems thus to be rather challenging for human inspection.
As we saw already for the proof of~\cref{thm:fallacy} that a hand written case distinction will be too tedious we rely here on the program.
It generated a total of $400$ shortest flip sequences which is precisely $\binom{6}{3}\cdot\binom{6}{3}$, that is, the number of different orders how the perfect flips between parking flips can interleave.

We even took the construction one step further and provide a pair of trees in which for every shortest flip sequence there exists an edge that has trace length four, that is, it is two times reparked. The example is depicted in~\cref{fig:manyparking}. The edges in red are part of $T_I$, edges in blue are part of $T_F$, and the black edges on the convex hull are contained in both trees. We verified our example with the same program used for the two previous counterexamples.

\begin{figure}[h]
    \centering
    \includegraphics[page=9]{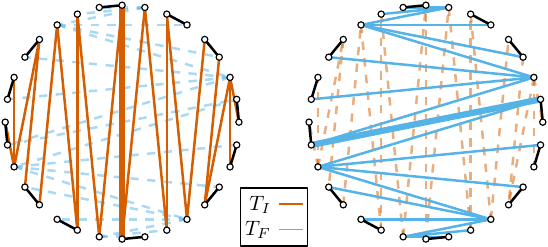}
    \caption{An example with 32 vertices in which every shortest flip sequence contains an edge that has trace length four (highlighted as the bold edge).}
    \label{fig:manyparking}
\end{figure}


\section{Conclusions and future work}
\label{sec:conclusions}

We established counterexamples to the parking and reparking properties for shortest flip sequences in the general setting, while proving that the reparking property holds when flips are restricted to be compatible.
These results clarify the limits of existing assumptions and point to several promising directions for future research.
\medskip
\begin{itemize}
    \item In~\cref{tab:algorithm-table} we summarized all recent upper bounds on the diameter and specify which of the investigated properties each bound respects. 
   	Most notably, in every case the trace length of all edges is bounded from above by~$2$. This raises two questions:
    Will our constructions need to be incorporated to achieve further improvements in the upper bounds on the diameter of the flip graph, and do these counterexamples offer new insights for constructing tighter lower bounds?
    \item By disproving the parking conjecture for general flips, we have eliminated a previously promising approach of establishing the happy edge conjecture.
    Consequently, we think that entirely new methods are required to make further progress on that conjecture.
    \item We provided a construction of initial and final trees where every shortest flip sequence for these pairs has an edge of trace length $3$ or $4$, respectively. 
    We expect that this construction can be iterated without limit. 
    Thus, we raise the following conjecture: 
    \end{itemize}
    \begin{conjecture}
    	For any constant $c$ there exists a pair of trees $T_I$ and $T_F$ such that any shortest flip sequence between $T_I$ and $T_F$ has an edge with trace length larger than $c$.
    \end{conjecture}
    \begin{itemize}
    \item Very recently it was shown that computing shortest (compatible) flip sequences between plane spanning trees of convex point sets is \NP-complete~\cite{bjerkevik2026flippingHardness} by Bjerkevik, Dorfer, Kleist, Ueckerdt, and Vogtenhuber. 
    This highlights the importance of identifying effective \FPT- or approximation algorithms. 
    Can the reparking property for compatible flips be exploited to enhance the currently known \FPT-algorithm of~\cite{AichholzerDV25}?
    Moreover, our counterexamples to the parking and reparking properties may be instrumental in establishing higher intractability results.
\end{itemize}

    \bibliography{references}
	\newpage
    \appendix
    \section{Omitted details from~\cref{sec:reparking-compatible-flips}}

\lemmaPreHappy*
\begin{proof}
    Without loss of generality, let $e\to f$ be the first and $f\to e'$ the last flip of the flip sequence. Further let $e \to f$ be a compatible flip, otherwise look at the reversed flip sequence. We denote with $v_a$ and $v_b$ the two vertices of $f$.
		
	We reorder and modify the other flips based on the following observation. The edge~$f$ separates the point set into two sides $A$ and $B$, the points that appear in clockwise direction from $v_a$ to $v_b$ and the points that appear in counterclockwise order, including $v_a$ and~$v_b$ in both $A$ and $B$. While $f$ is part of the tree, flips in the two sides can be executed independently from the other side. Note that every spanning tree contains a unique path between any two vertices. If, in particular, we consider the path~$p$ from~$v_a$ to~$v_b$ in~$T_I=T_{0}$, we see that this path has to be contained entirely in one side induced by~$f$, w.l.o.g.\ in side~$B$. Otherwise, there would be an edge along the path that crosses~$f$, introducing a non-compatible flip from~$T_I$ to~$T_{1}$.
		
	If we add the edge~$f$, the cycle that gets closed is entirely contained in~$B$. Therefore, $e$ lies in~$B$. Any flip that is executed from~$T_{1}$ to~$T_{k-1}$ that happens in~$A$ can already be executed in~$T_{I}$. Any cycle that previously involved the edge~$f$ now involves the path~$p$ instead. We order our flips in a way, such that all flips that happen in~$A$ are executed right away starting from~$T_{I}$. After executing all those steps, the side~$A$ already coincides with the side~$A$ in~$T_{F}$, except for the edge~$f$.
		
	Let us take a closer look at the part of~$T_{F}$ that lies in $A$. This part consists of $f$ and two trees~$T$ and $T'$ (possibly with an empty set of edges) in~$A$. Since~$T_{F}$ is plane there exist a unique pair of points that appear consecutively along the boundary of the convex hull and one belongs to~$T$ and the other to~$T'$. We continue our modified flip sequence by adding the convex hull edge~$h$ between those two points and removing the edge that would have been removed when adding~$f$.
		
	Now we execute all the flips in~$B$ from the original flip sequence. Any flip that formed a cycle that involves~$f$ now forms a cycle that involves~$h$ and paths in~$A$ connecting~$v_a$ and~$v_b$ to~$h$.
		
	As a last flip we perform $h\to e'$ and reach the final tree $T_F$. The trace of $e$ is $(e,h,e')$ and for any other edge in $T_I$ the trace did not change between the two flip sequences.
\end{proof}

\section{Omitted details from~\cref{sec:counterexamples}}

\lemmaParkingEdgeProperty*
\begin{proof}
    Consider the pair of trees~$T_I$, $T_F$ depicted in~\cref{fig:parking-counterexample}.
    In~\cref{fig:parking-counterexample-opt}, a sequence using a diagonal parking edge~$p_1$ (green) needs exactly one parking flip, $e_1 \to p_1$, while all other flips are final.
    Further, note that the first flip of any sequence from~$T_I$ to~$T_F$ must be a parking flip since every edge in the symmetric difference~$T_I\Delta T_F$ is crossed at least twice.
    Hence the sequence depicted in \cref{fig:parking-counterexample-opt} is a shortest flip sequence.
    We note that this flip sequence does also contain two edge slides. These are highlighted in purple.

    To prove the statement, it remains to show that any parking flip that parks to the convex hull instead of~$p_1$ requires at least one more flip.
    To do so, we show that no flip sequence that starts with a parking flip to the convex hull can reach~$T_F$ by performing only final flips after this parking flip.
    In~\crefrange{fig:parking-counterexample-park-1}{fig:parking-counterexample-park-4}, we hint at all the possible first flips that remove one diagonal of~$T_I\setminus T_F$ add an edge on the convex hull instead. In every case we depict one flip sequence that performs perfect flips for as long as possible until no other edge of~$T_F$ can be added without performing another parking flip. We will refer to such flip sequences as \emph{maximal} flip sequences.

    The flip sequences in~\cref{fig:cases} do not cover all cases that could possibly occur, since there are sometimes multiple edges that can be flipped to the same edge or multiple flips that can be executed in a different order while leading to the same result. For the case where the first flip is~${e_2\to g_8}$ a graph depicting all possible maximal flip sequences can be seen in~\cref{fig:example-parking-to-gap-8}. All the other cases in which a diagonal is parked on the convex hull can be argued similarly. 
    Since already the first case highlights how tedious the resulting case distinction is, we will resort to two different approaches to cover all cases. All the arguments in the case distinction boil down to three basic arguments: 
    (1)~If an edge of~$T_F$ is crossed at least twice, then it cannot be added by a single flip. 
    (2)~If an edge of~$T_F$ is crossed exactly once, then the only flip that could potentially add that edge is the flip that removes the edge that crosses it.
  (3)~We cannot add edges of~$T_F$ if adding the edge would create a cycle in the graph. In all cases, the cycle will contain a parking edge on the convex hull.

    After the proof, we will discuss how we checked all cases via a simple program~\cite{visual-and-verifier} and provide a directed acyclic graph~\cite{visual-and-verifier} that embeds all flip sequences that can be made with~$8$ flips where at most one of the flips can be a parking flip.

    Back to the flip sequences starting with~${e_2\to g_8}$. On the top, the initial tree is depicted, followed by the tree resulting from the flip~${e_2\to g_8}$. There are only two edges of~$T_F$ that are not crossed twice, namely~$f_4$ and~$f_3$, which are both crossed by the edge~$e_2$. Flipping~${e_1\to f_3}$ will disconnect the tree so the only valid flip is~${e_1\to f_4}$.

    For the next flip, we have multiple edges of~$T_F$ that can be inserted, since~$f_3$ is not crossed at all and the edges~$\overline{f_4}$, and~$\overline{f_3}$ are only crossed by~$\overline{e_2}$. To insert~$f_3$ the three flips~${e_3\to f_3}$,~${\overline{e_2}\to f_3}$, and~${g_8\to f_3}$ are possible. To insert~$\overline{f_3}$, we can perform the flip~${\overline{e_2}\to f_3}$. Note that no flip can add the edge~$\overline{f_4}$ since~$\overline{f_4}$ is crossed by~$\overline{e_2}$ and flipping~${\overline{e_2}\to \overline{f_4}}$ creates a cycle. 

    In the left tree of the fourth row, there are three edges in~$T_F$ that are crossed by a single edge of~$T_I$. More precisely, $f_2$ is only crossed by~$e_4$ and~$\overline{f_3}$ is only crossed by~$\overline{e_2}$. Similarly, also~$\overline{f_4}$ is only crossed by~$\overline{e_2}$, but directly flipping~${\overline{e_2}\to\overline{f_4}}$ creates a cycle. We can perform the flips~${e_4 \to f_2}$ and~${\overline{e_2} \to \overline{f_3}}$ in some order. After the flip~${\overline{e_2} \to \overline{f_3}}$, also the flip~${g_8\to\overline{f_4}}$ becomes possible.
    This gives a total of three different flip sequences that perform the flips~${e_4 \to f_2}$,~${\overline{e_2} \to \overline{f_3}}$, and~${g_8 \to \overline{f_4}}$, respecting that the last flip has to occur after the middle flip.
    All result in the same situation with two edges of~$T_F$ being crossed twice by two edges~of~$T_I$.

    In the second tree in the fourth row, we have so far added the edges~$f_4$, $f_3$ and~$g_8$. The only edges of~$T_F$ that are not crossed twice are~$\overline{f_4}$ and~$\overline{f_3}$. $\overline{f_4}$ can only be added via the flip~${g_8\to\overline{f_4}}$ without creating a cycle. $\overline{f_3}$ can be added via the flips~${e_3\to\overline{f_3}}$ and~$\overline{e_3}$ to~$\overline{f_3}$.

    By flipping~${g_8\to \overline{f_4}}$, we can either continue with flipping~${\overline{e_3}\to\overline{f_3}}$ then~${\overline{e_4} \to \overline{f_2}}$ or with~${e_3\to \overline{f_3}}$ and~${e_4 \to f_3}$. In both cases, we end up with two edges from~$T_I$ being intersected by two edges of~$T_F$. We cannot add any further edge of~$T_F$ without performing another parking flip.

    By flipping~${e_3\to \overline{f_3}}$ we are left with two more possible flips, namely~${g_8 \to \overline{f_4}}$ and~${e_4\to f_3}$. Any order of the two flips yields a valid flip sequence. Both times we are then left with a situation in which both~$\overline{f_4}$ and~$f_1$ from~$T_F$ are crossed by~$\overline{e_3}$ and~$\overline{e_4}$ from~$T_I$. We cannot add any further edge of~$T_F$ without performing any further parking flips. A symmetric situation occurs after flipping~${\overline{e_3}\to \overline{f_3}}$ with the two flips~${g_8\to \overline{f_4}}$ and~${\overline{e_4}\to \overline{f_2}}$.

    We next consider the third tree in the fourth row. We remark that the edge~$\overline{f_3}$ splits the tree into two parts. Flips in the two parts can be shuffled arbitrarily. This yields multiple flip sequences that have the same outcome and only differ by the ordering of the flips, more particularly, when the flip~${g_8\to \overline{f_4}}$ occurs. In the side of~$\overline{f_3}$ that does not contain~$g_8$ there are two flip sequences. Those are~${e_3\to f_3}$ followed by~${e_4 \to f_2}$ and~${\overline{e_3} \to f_3}$ followed by~${\overline{e_4} \to \overline{f_2}}$. Both end in the situation that two edges of~$T_F$ are crossed twice by two edges of~$T_I$.

    In the last tree in the fourth row, we have added the edges~$f_3$ and~$f_4$ of~$T_F$. The only edges of~$T_F$ that are not crossed twice are~$\overline{f_3}$ and~$\overline{f_4}$ which are both crossed once by~$\overline{e_2}$. The only possible non-parking flip is~${\overline{e_2}\to \overline{f_4}}$. Now the only edge that can be is not crossed twice is~$\overline{f_3}$ which can be added by the two flips~${e_3 \to \overline{f_3}}$ and~${\overline{e_3}\to \overline{f_3}}$. ${e_3 \to \overline{f_3}}$ can only be followed by~${e_2\to f_2}$, whereas~${\overline{e_3}\to\overline{f_3}}$ can only be followed by~${\overline{e_2} \to \overline{f_2}}$. In both cases we again end up with two edges of~$T_F$ being crossed by two edges of~$T_I$.
\end{proof}

\section{Additional figures}
\label{app:figures}

\begin{figure}[htp]
    \centering%
    \includegraphics{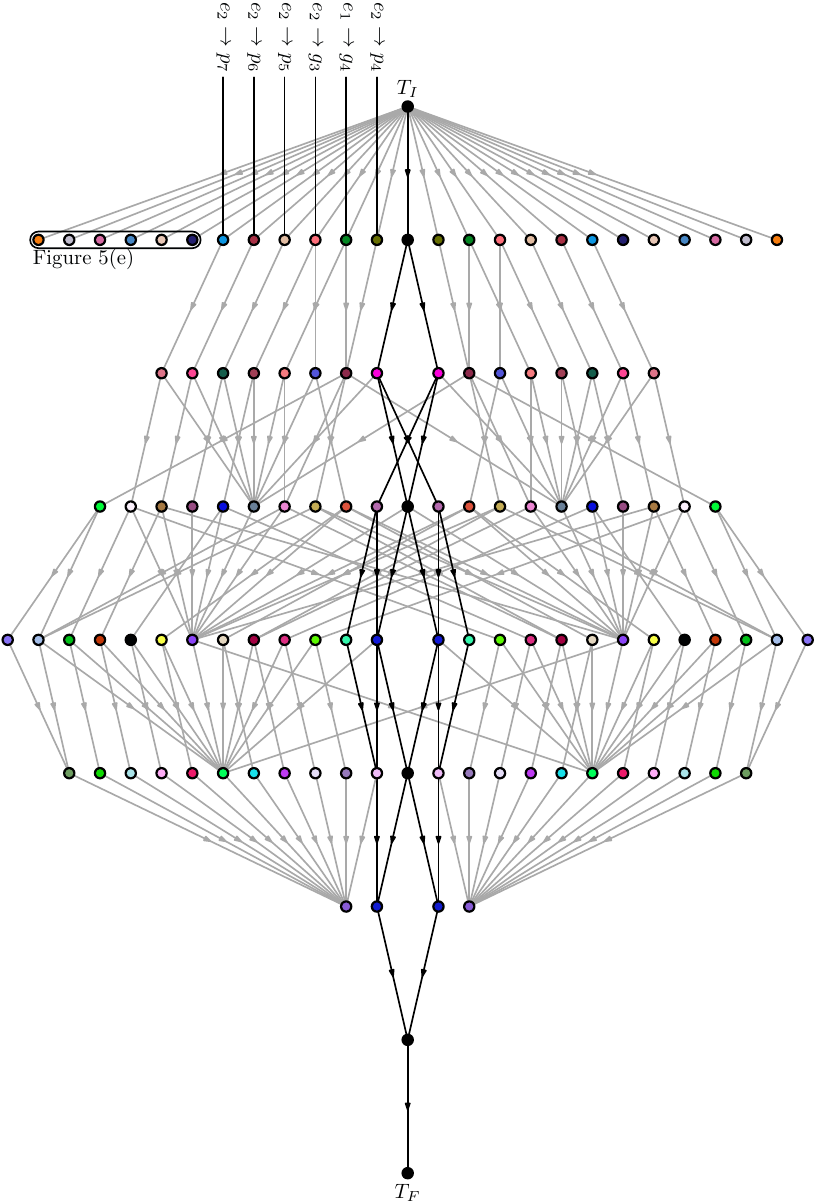}%
    \caption{
        Reachable subgraph of trees from~$T_I$ as shown in~\cref{fig:parking-counterexample}, when making at most one parking flip.
        Black edges indicate that performing this flip will not make the target tree $T_F$ unreachable by the remaining sequence, and node colors indicate symmetric trees.
    }
    \label{fig:counterexample-dag}
\end{figure}

\begin{figure}[p]
    \centering%
    \includegraphics{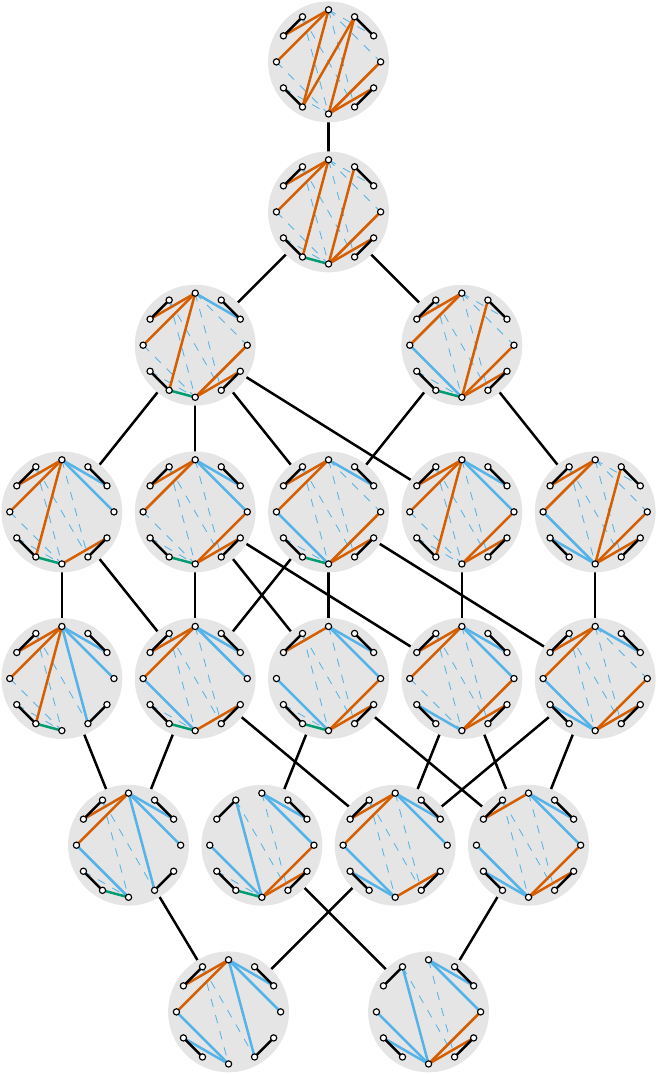}%
    \caption{Reachable subgraph of trees after flipping~${e_1\to g_4}$.}
    \label{fig:descends_383}
\end{figure}

\begin{figure}[p]
    \centering%
    \includegraphics{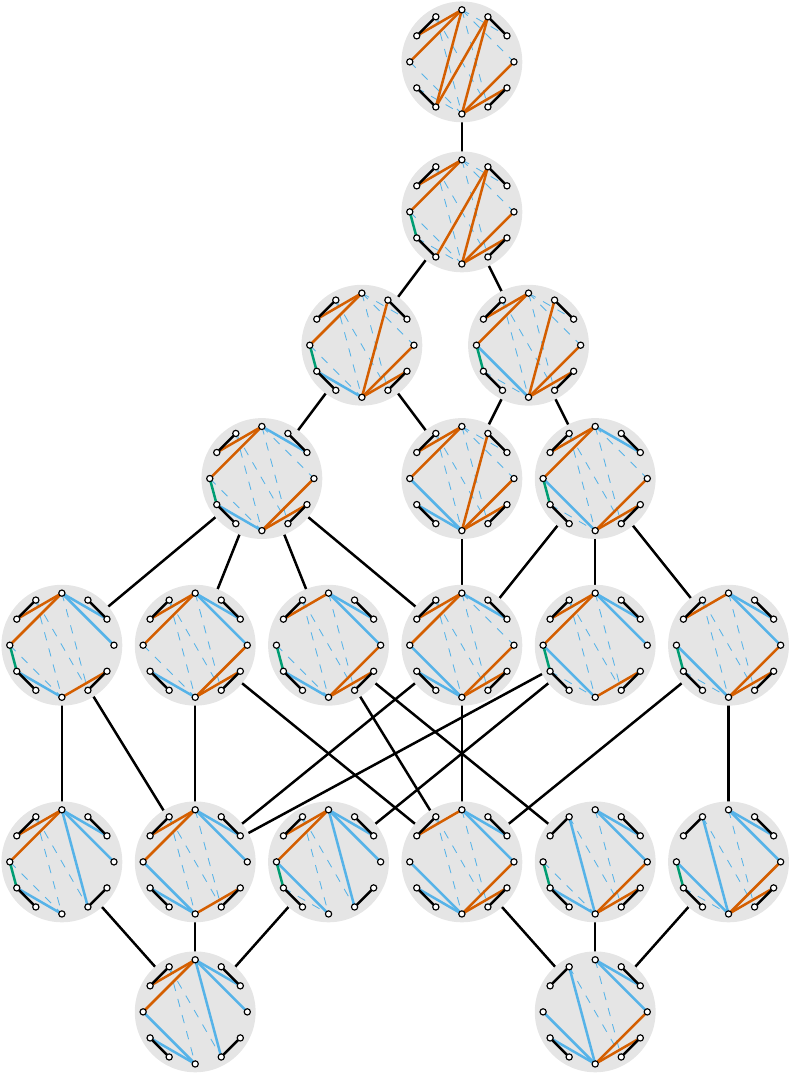}%
    \caption{Reachable subgraph of trees after flipping~${{e_2}\to g_3}$.}
    \label{fig:descends_80}
\end{figure}

\begin{figure}[p]
    \begin{subfigure}[t]{0.45\columnwidth - 0.5em}%
        \centering%
        \includegraphics{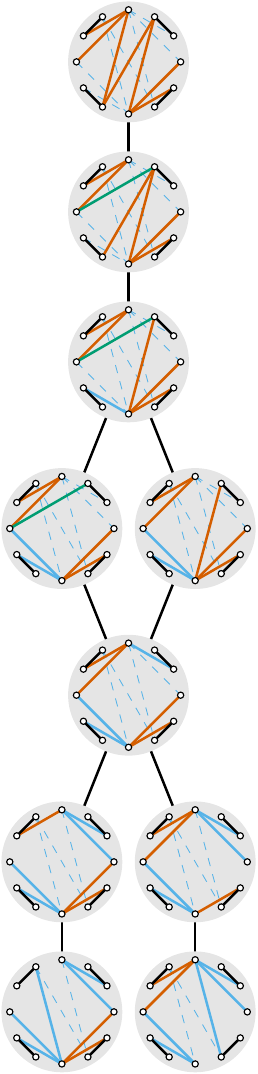}%
        \subcaption{Reachable subgraph of trees after flipping~${e_2\to p_5}$, where ${p_5=v_{4}v_{12}}$.}
        \label{fig:descends_156}%
    \end{subfigure}%
    \hfill%
    \begin{subfigure}[t]{0.55\columnwidth - 0.5em}%
        \centering%
        \includegraphics{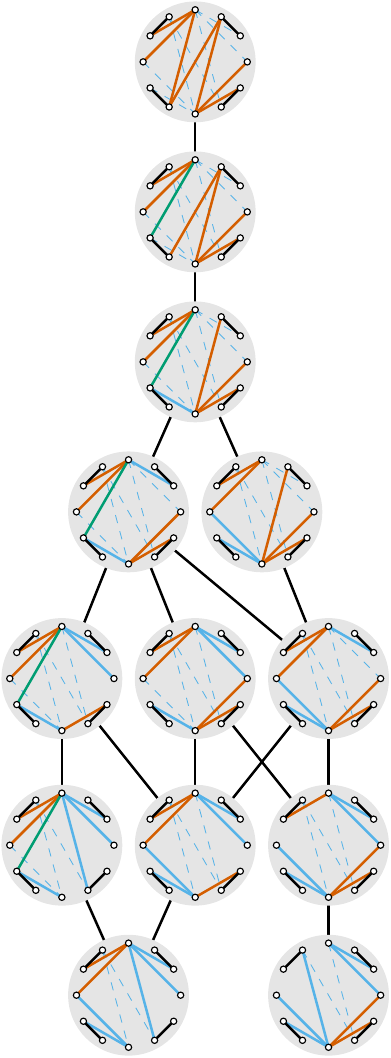}%
        \subcaption{Reachable subgraph of trees after flipping~${{e_2}\to p_6}$, where~${p_6=v_{1}v_{5}}$.}
        \label{fig:descends_7}%
    \end{subfigure}%
    \caption{The subgraph of reachable trees after two diagonal park flips. These are not relevant to the correctness of~\cref{thm:fallacy} but included for completeness.}
    \label{fig:descends_156_7}%
\end{figure}

\begin{figure}[p]
    \centering%
    \includegraphics{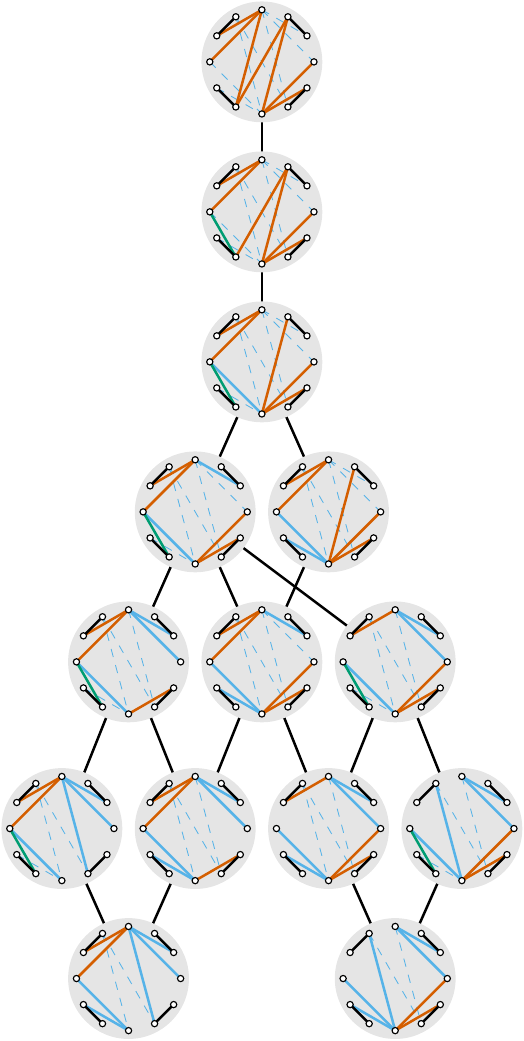}%
    \caption{Reachable subgraph of trees after flipping~${e_2\to p_7}$, where $p_7=v_{4}v_{6}$.
    This is not relevant to the correctness of~\cref{thm:fallacy} but included for completeness.}
    \label{fig:descends_132}
\end{figure}

\end{document}